\documentclass{ipi}
\usepackage{amssymb}
\usepackage{graphicx}
  \usepackage{amsmath}
  \usepackage{paralist}
  \usepackage{graphics} %% add this and next lines if pictures should be in esp format
  \usepackage{epsfig} %For pictures: screened artwork should be set up with an 85 or 100 line screen

  \usepackage{hyperref} %% Warning: when you first run your tex file, some errors might occur, please just
\def\currentvolume{X}
 \def\currentissue{X}
  \def\currentyear{200X}
   
    \def\ppages{X--XX}

\newtheorem{theorem}{Theorem}[section]

\newtheorem{lemma}[theorem]{Lemma}

\newtheorem*{problem}{Problem}
\theoremstyle{definition}
\newtheorem{definition}[theorem]{Definition}
\newtheorem{remark}{Remark}

\def\be{\begin{equation}}
\def\ee{\end{equation}}
\def\bee{\begin{equation*}}
\def\eee{\end{equation*}}
\def\bea{\begin{eqnarray}}
\def\eea{\end{eqnarray}}
\def\p{\partial}

\newcommand{\R}{{\mathbb R}}
\newcommand{\E}{{\mathcal E}}

\title[Dynamical tomography]
{Dynamical tomography of gravitationally bound systems}

\author[Mikko Kaasalainen]{}

\thanks{Supported by the Academy of Finland, project "New mathematical methods in planetary
and galactic research"}

\begin{document}
\maketitle

\centerline{\scshape Mikko Kaasalainen}
\medskip
{\footnotesize
 %% please put the address of the first author
 \centerline{Department of Mathematics and Statistics}
   \centerline{P.O. Box 68, FI-00014 University of Helsinki}
   \centerline{Finland}
} 

\bigskip

%% The name of the associate editor will be entered by an editorial staff
 \centerline{(Communicated by Jari Kaipio)}

\begin{abstract}
We study the inverse problem of deducing the dynamical characteristics 
(such as the potential field) of large systems from kinematic observations.
We show that, for a class of steady-state systems, 
the solution is unique even with fragmentary data, dark matter, 
or selection (bias) functions. Using spherically symmetric
models for simulations, we investigate solution convergence and 
the roles of data noise and 
regularization in the inverse problem. We also present a method,
analogous to tomography,
for comparing the observed data with a model probability 
distribution such that the latter can be determined.
\end{abstract}

%\newpage
\section{Introduction}

Matter in the universe is usually contained in
systems bound together by gravitation: planets and their moons,
planetary systems, star clusters, galaxies, and groups of galaxies. Indeed, the
modern concepts of gravitation and gravitational potential were brought
about by the realization that a mathematically well definable universal
force field must keep celestial bodies in their observed orbits. Newton's
solution to the inverse problem of ``What kind of a force keeps two 
pointlike bodies in an elliptic orbit around each other?'' was the 
inverse-square law of 
attraction\footnote{Newton actually solved
the direct problem of ``What kind of an orbit is produced by the inverse-square
attraction?'' (this is what Halley asked him). The solution of the inverse
problem is an instant corollary; it is also unique for motion around the 
focal point of an ellipse -- the harmonic oscillator creates an elliptic
orbit around the centre.}. 
In such a force field, the equations of motion for a body can be written as
\bee
\frac{d^2 x}{dt^2}=-\nabla \Phi(x),\label{newton}
\eee
where $t\in\R$ is the time, $x\in \R^3$ describes the position of the body,
and $\Phi:\R^3\rightarrow \R$ is the gravitational potential; here
$\Phi(x)\sim \Vert x\Vert^{-1}$.

Outside isolated two-body configurations, 
the two-body solution only serves as a useful first 
approximation for systems dominated by one massive body; 
for such a class of systems,
the perturbations from this solution can be examined analytically
to a relatively high precision. The behaviour of a system consisting of 
a moderate number of bodies of more or less equal weight must be numerically 
integrated assuming the Newtonian potential for each pointlike mass. However,
if the system consists of many ($>10^6$ or so) gravitationally interacting 
bodies collectively forming the potential, we can
assume it to obey some principles of statistical mechanics. Accordingly,
we no longer analyze the motion of separate particles, but want
to define the dynamical characteristics of the system as a whole, such as its
global-scale potential field  $\Phi(x)$, matter density $\rho(x)$, 
etc. We define {\em dynamical tomography} as the inverse problem of deducing these smoothly distributed characteristics from kinematic observations. A more detailed definition of this large-scale generalization of Newton's
dynamical inverse problem is given in Section 2.1. The term tomography can be used here in both
intuitive and technical senses: we determine density/distribution functions inside a domain, and
this can be done by employing the principles of tomographic analysis.

Large gravitationally bound systems have been studied for more
than a century, and the corresponding theoretical framework is described 
in several textbooks, of which \cite{bt}
is the most up-to-date one. However, the studies have mostly concentrated 
on the direct problem of creating numerically or analytically well
tractable dynamical models, or of
designing approximate system models to
explain various observed phenomena, while the
general inverse problem (as posed in sections 2 and 3 here)
has received less attention. This is mostly due
to the lack of data sufficient for a comprehensive analysis, and the 
scarcity of efficient methods of dealing with such data. New sky surveys
such as LSST (Large Synoptic Survey Telescope, USA) and Gaia (EU), both
starting at the beginning of the next decade, 
will provide a vast amount of kinematic data of the stars in our galaxy, making the
construction of a robust machinery for inverse problem analysis essential.
We describe here some theoretical and practical aspects of solving problems 
in dynamical tomography.

The structure of the paper is as follows: in Section 2, we review as well
as define some
basic concepts of large gravitational systems, and pose the problem in 
corresponding terminology. 
In Section 3, we examine the fundamental aspects and uniqueness
properties of the inverse problem, and in Section 4 present examples
in the simplified case of spherically symmetric systems. 
In Section 5, we redefine the problem in
the sense of probability distributions, which is usually necessary in practice.
Finally, in Section 6, we present conclusions and discuss future work
and applications.

\section{Defining large systems}

\subsection{Distribution functions}

In this paper, we assume the system to be bound, stable, collisionless 
($d\gg D$, where $d$ is the shortest distance between two bodies of size $D$)
and in a steady-state equilibrium (for a discussion on these standard 
assumptions and their applicability regimes, 
see \cite{bt}). In practice, this is a good 
approximation for large stellar systems, much in the same way as 
the two-body solution is a good starting point for analyzing the solar system.
The dynamics of such a system is completely described by
the smooth phase-space distribution function $f(x,v,t)$ of its constituent
particles (stars), $f:\R_x^3\times \R_v^3
\times \R_t\rightarrow \R$,
where $v\in\R^3$ is the velocity. The function $f$ pertains to the limits 
$d\rightarrow 0$ (infinitely many particles), $D\rightarrow 0$, and 
$D/d\rightarrow 0$. In practice, $f$ is interpreted as the quantity
giving the expected number of particles in a phase-space volume via
$f\,d^3x\,d^3v$, or as the probability density
of finding a particle at $(x,v)$.

The distribution function $f$ essentially 
corresponds to fluid in six dimensions; the
equation of motion for $f$ is the collisionless Boltzmann (or Vlasov) equation
\be
\frac{\p f}{\p t}+\langle v,\nabla_x f\rangle-\langle\nabla\Phi,\nabla_v f\rangle=0\label{bol}
\ee
(a special case of Liouville's theorem; see \cite{bt}). In terms of Poisson
brackets, this reads $\p f/\p t=[H,f]$, where $H$ is the Hamiltonian.
This implies that the flow
of the fluid (the smooth motion of particle phase points through phase space)
is incompressible. In this paper, we study the steady-state case where
matter units move in their orbits in the potential field $\Phi$ while 
leaving the collective system configuration completely unchanged:
$\p f/\p t=0=\p\Phi/\p t$, i.e., $f=f(x,v)$ and $\Phi=\Phi(x)$. Moreover, the
observational data, measurements of $(x,v)$ for stars or equivalent units
of matter, are effectively for one epoch: though they are
measured over several years, the orbital periods are millions of years,
so the data do not provide orbital information outside the
linearized regime $v=d x/d t$.

We also define a set 
$Q$ including all the other characteristics than $x,v$ that can be assigned
to a star from observations, such as luminosity, temperature, 
chemical composition, etc. These may be direct observables or implicit quantities 
such as mass and age that can be expressed in terms of direct ones.
They can be continuous or even discrete (taxonomy), in which case an
integral over such a variable is understood in a suitable corresponding sense.
Thus we have a full distribution function $F(x,v;Q)$ of $N+6$ dimensions,
where $N$ is the number of the quantities in the set $Q$. At first sight, one
might imagine that only mass could be a dynamically interesting quantity
in $Q$, but, e.g., luminosity directly affects the observability and bias
factors. Also, as will be shown below, the possibility to partition 
the observed
particles into different populations with the aid of $Q$
adds to the information content.

We can distribute $f(x,v)$ further among the members of $Q$ at each $(x,v)$
with a function $\varphi(x,v;Q)$:

\begin{definition} The full distribution function of number
density in $\R_x^3\times \R_v^3$ and in the set $Q$ of $N$ quantities $q$
(generic notation for a member of $Q$) is
\be
F(x,v;Q)=f(x,v)\varphi(x,v;Q),
\ee
normalized such that $\int_Q \varphi(x,v;Q)\, d^N q=1$.
The marginal distribution of some $q\in Q$, $q_0\le q\le q_1$ is
\be
F(x,v;q)=f(x,v)\varphi(x,v;q)=
f(x,v)\int_{Q\setminus q} \varphi (x,v;Q)\, d^{N-1}q'
\ee
and $\int_{q_0}^{q_1} \varphi (x,v;q)\, dq=1.$
\end{definition}

\begin{definition} The distribution function $f_i(x,v)$ of
an object population $P_i$ is given by
\be
f_i(x,v)=f(x,v)\mathcal{P}_i(x,v),
\ee
where the population filter $\mathcal{P}_i(x,v)$ is
\be
\mathcal{P}_i(x,v)=\int_{\Lambda_i} \varphi(x,v;Q)\, d^N q,
\ee 
where $\Lambda_i$ sets the (integration) limits for each observable of
 $Q$ according to the given characteristics of $P_i$.
 \end{definition}

An important distribution function is that of mass per $d^3x\,d^3v$, 
denoted by $\tilde f(x,v)$:

\begin{definition} The mass distribution function is given by
\be
\tilde f(x,v)=f(x,v)\mathcal{M}(x,v),
\ee
where the mass factor $\mathcal{M}(x,v)$ is
\be
\mathcal{M}(x,v)=\int_0^\infty M \varphi (x,v;M)\, dM.
\ee
\end{definition}

For example, the case of all particles having the same mass $M_0$ is given
by
\bee
\varphi(x,v;M)=\delta(M-M_0),
\eee
where $\delta$ is the Dirac delta distribution. Thus $\mathcal{M}=M_0$.

The amount of matter in an infinitesimal volume of phase space 
$\R_x^3\times \R_v^3$ is
$\tilde f\, d^3x\,d^3v$, and the matter density $\rho(x)$ in $\R_x^3$ is
\be
\rho(x)=\int_{V_+} \tilde f(x,v) \,d^3v,\label{rho}
\ee
where the velocity domain $V_+$ includes all velocities $v$ that can exist
at $x$ in a bound system; alternatively, for $v\notin V_+$ we define $f(x,v)=0$
and can integrate over all $v$. The matter density and the potential are
related through Poisson's equation:
\be
\rho(x)=\frac{1}{4\pi G}\nabla^2\Phi(x),\label{poisson}
\ee
where $G$ is the universal gravitation constant. 
Note that if $\tilde f$ is the 
distribution of all matter, equations (\ref{bol}), (\ref{rho}) and 
(\ref{poisson}) must be simultaneously fulfilled. The search for such
self-consistent solution pairs of $\tilde f,\Phi$ is the fundamental problem 
of stellar dynamics. However, as we will show in Section 3, 
$f$ (or $\tilde f$) does not need to cover all matter in
our problem; it can describe a selected population, or there may be 
unobservable matter such that the full $f$ cannot be determined while 
$\Phi$ can.

The distribution function $f$ is a statistical tool, meant to
describe the average distribution of matter in large enough bins. As 
mentioned above, another
useful practical interpretation of $f$ is the probability density
of observing a star (or, for $F(x,v;Q)$, a star with properties Q) 
at $(x,v)$; this circumvents the local-scale distribution of matter 
in space, and the size of our averaging bins -- the collisionless 
approximation means (and observations show) that most of space is void of 
luminous matter. In section 5, we will use this interpretation to redefine
the observables of our inverse problem.

We can now state our problem of dynamical tomography as follows:

\begin{problem}
Given a large number of observed $(x,v)$ (for any motion markers
such as stars or other matter and possibly in different populations
$P_i$) in a domain $\Omega\subset \R^3\times \R^3$
in a gravitationally bound steady-state system, deduce the
potential $\Phi(x)$ of the system, and the distribution function(s) $f(x,v)$
of the observed matter in $\Omega$.
\end{problem}

Obviously, if the number of observations is
extremely large and all matter in the system is 
observed, the problem is trivial. The number of observations in arbitrarily
small phase-space volumes is now nonzero, so an accurate estimate of
the distribution $f(x,v)$ is obtained
by simply counting the objects, assuming that their masses are known. From
this, $\Phi(x)$ follows; one could, as a shortcut to $\Phi(x)$, directly
count objects in $\R_x^3$ only. 
What makes the problem an inverse one is that
neither of the requisites is fulfilled in reality: the number of 
observations, though high, is not sufficient for direct counting in small
enough phase-space (or even $\R_x^3$) bins, and, above all, not all 
of the system matter is observed because of bias factors or non-luminous
matter. This is why we need to consider 
$\R_x^3\times\R_v^3$ instead of just $\R_x^3$ even if we only want to
obtain $\Phi(x)$: the fraction of unobserved matter can only be inferred by
using a priori information furnished by the expected dynamical properties
of the system. 
Our goal is now to find a parametrization in which the dynamics
and a priori constraints of the system can be expressed in a practical 
and preferably analytically tractable way.
Here we neglect effects such as spatial
correlation between stars, 
changing mass (shedding or accretion) due to stellar evolution, etc. 
\cite{bt}. 

Some basic properties of the system can already be deduced from a very 
small number of observations. For example, if we observe matter
at distance $r$ from the centre of the system, moving at speed $v_r$ relative
to us,
and assume it to be bound by the system, we can roughly estimate the
mass $M$ inside $r$, e.g., by assuming a spherical mass distribution
and a circular orbit around the centre: $M=v_r^2r/G$.
In this and similar ways it has been deduced from several observations of 
galaxies and their clusters
that most of the mass in the universe is contained in dark matter:
non-luminous and extremely hard if not impossible to observe directly.

\subsection{Jeans theorems and integrable systems}

As is well known, a natural and compact way of parametrizing the 
desired dynamical properties of the system can be achieved via
its dynamical invariants.
A function $I(x,v)$, evaluated at $[x(t),v(t)]$ of
any orbit in the potential $\Phi$ of the system, is an {\em integral of motion} if
\bee
\frac{d}{d t} I[x(t),v(t)]=0
\eee 
for all orbits at all $t$. Evaluating the time derivative by the chain rule
and substituting the equation of motion (\ref{newton}), we obtain 
an equation exactly of the form (\ref{bol}) ($I$ corresponding to $f$)
with the steady-state condition $\p f/\p t=0$.
Thus we arrive at the Jeans theorem (see, e.g., \cite{bt}) 
that states that any steady-state
solution $f$ of (\ref{bol}) is of the form $f[I_i(x,v)]$, where $I_i$
are integrals of motion.
For example, the energy $E(x,v)$:
\be
E(x,v)=\frac{1}{2}\langle v,v\rangle+\Phi(x) \label{ene}
\ee
is always an integral in gravitationally interacting systems.

To have a complete set of integrals which we can, in principle, 
reconstruct and thus properly use in our formulation of the problem, we 
further restrict ourselves to regular orbits and integrable systems (defined below), 
a subset of the steady-state solutions of (\ref{bol}). 
This class of solutions, though restricted, has two great advantages:
\begin{enumerate}[i)]

\item it facilitates an analytical or semianalytical investigation of the 
problem as well as the derivation of uniqueness results; and

\item it can be readily viewed as a further modifiable approximate
representation of the full class of solutions.

\end{enumerate}

An orbit is regular (quasiperiodic) if it is confined
to a 3-torus $S^1\times S^1\times S^1$ in phase space \cite{arn}. Then
it has three integrals (also called isolating integrals)
$I_i, i=1,2,3$ that can be used to define the
torus: $I_i(x,v)=C_i$, where $C_i$ are constants. The isolating integrals
can be given in an arbitrary basis as functions of $I_i$ are isolating
integrals as well. A proper basis set of isolating integrals
maps each torus $\mathcal T$ as ${\mathcal T}\leftrightarrow I_i, i=1,2,3$, 
i.e., any such bases
can be mapped one-to-one into each other: $\lbrace I'_i\rbrace\leftrightarrow 
\lbrace I_i \rbrace$.

If a potential $\Phi(x), x\in \R^3$ creates an integrable system, all orbits 
are confined to 3-tori in $\R_x^3\times \R_v^3$
that are each defined by three action (Poincar\'e)
integrals $J_i, i=1,2,3$, a class of isolating integrals $I$:
\be
J_i=\frac{1}{2\pi}\oint_{\mathfrak{P}_i} \langle p, dq\rangle,\label{action}
\ee
where $p\in \R^3$ and $q\in \R^3$ are any canonically conjugate
momenta and coordinates, and 
$\mathfrak{P}_i$ is a path that cannot be continuously deformed into a point.
For other paths, the integral vanishes. There are three such sets of 
possible paths each producing one $J_i$; see, e.g., \cite{arn,bt}.
$J\in \R_J^3$ plays the role of canonical momentum in Hamilton's equations,
and its canonically conjugate coordinate pair is the angle variable 
$\theta\in \R^3$. $(J,\theta)$ are related to phase-space coordinates
$(x,v)$ by a canonical transformation, and energy (Hamiltonian) depends
on $J$ only: $E=E(J)$.

Action-angle formalism yields the time averages theorem \cite{bt} 
that states that the density of the orbit on the torus, i.e., the average
time it spends on different parts 
of the torus, is evenly distributed in $\theta$ 
on its surface. This holds for non-resonant, i.e., not closed 
orbits that make up almost all of phase space. 
This, in turn, leads to the strong Jeans theorem
that essentially states that, in an integrable system,
a steady-state solution of (\ref{bol}) must be of the form 
\be
f(x,v)=f[I_1(x,v),\dots ,I_3(x,v)]:= f[I(x,v)], \label{jeans}
\ee
where $I_i$ are isolating integrals. 

If we want a steady-state solution for the luminous matter in the system,
the distribution function of our choice is the mass distribution
$\tilde f(x,v)$, i.e., we now assume that $\tilde f=\tilde f[I(x,v)]$.
Thus, if we want $f$ and $\tilde f$ to be interchangeable in our steady-state
equations and theorems to be derived in section 3, 
we must assume the mass factor $\mathcal M$ to be of the form 
$\mathcal{M}=\mathcal{M}[I(x,v)]$ as well, i.e.,
$\varphi(x,v;M)=\varphi[I(x,v);M]$. Leaving the masses unknown is the 
advantage in using $f$, but in practice one should solve the inverse
problem with $\tilde f$ as well by weighting the count of each observed point 
$(x_i,v_i)$ by the corresponding mass $M_i$, and then compare the results
(e.g., the tori obtained).

The fact that the system as a whole is a steady-state and integrable one
does not generally imply 
that a population filter $\mathcal{P}_i(x,v)$ is of the
form $\mathcal{P}_i[I(x,v)]$, or that, for any other $q$ than the mass $M$,
$\varphi(x,v;q)=\varphi[I(x,v);q]$. But if such filters (i.e., steady-state
populations) can be assumed to be identifiable
at least in some domain of $Q$ and $\R_x^3\times\R_v^3$, they
are very useful in the inverse problem, as will be shown in section 3.

\section{Inverse problem}

\subsection{Uniqueness}

By virtue of the strong Jeans theorem,
we now seek integrable systems that reproduce the observations. Since the
orbits of all particles, observed or not, are on tori, and $f(x,v)$
is of the form (\ref{jeans}), we know that each isosurface
$f(x,v)=const$ in phase space $\R_x^3\times \R_v^3$
of any steady-state $f(x,v)$, regardless of the group of objects it 
represents, entirely consists of 3-tori on which $I(x,v)=const$.
To show that these isosurfaces determine the potential $\Phi(x)$ of the
system, let us first present the following lemma:

\begin{lemma} Let the potential $\Phi(x)$ generate an integrable system,
and let $\mathcal T$ denote the corresponding set of 3-tori in 
$\R^3\times \R^3$.
Then $\Phi(x)$ is the only integrable potential (up to an additive constant) 
that creates any chosen subset $\widehat{\mathcal T}$
of arbitrarily small patches $\Gamma$ on any tori of $\mathcal T$ such that 
$\widehat{\mathcal T}$ covers
all of $\R_x^3$ (accessible to the system) in a connected manner.
\end{lemma}

\begin{proof} The value $E$ of the energy (\ref{ene}) is constant on each torus.
Let us choose a patch $\Gamma_0$ on which
$\Phi(x)=E_0-\frac{1}{2}\langle v,v\rangle$, where $E_0$ is
the energy on the corresponding torus. Thus $\Phi(x)$ is defined up to
an arbitrary constant for all $x$ of the torus. If we have another patch
$\Gamma_1$ 
with a point $x=x_0$ common with $\Gamma_0$, and the corresponding
phase-space points on the two patches are, respectively, $(x_0,v_0)$ and
$(x_0,v')$,
the value $E_1$ of the energy on $\Gamma_1$ must be 
\bee
E_1=\Phi(x_0)+\frac{1}{2}\langle v',v'\rangle=E_0+\frac{1}{2}(\langle v',v'\rangle-\langle v_0,v_0\rangle). 
\eee
This defines $\Phi(x)$ on $\Gamma_1$: $\Phi(x)=E_1-\frac{1}{2}\langle v,v\rangle$. 
This chain of patches can be continued to cover $\Phi(x)$ everywhere (up to the
arbitrary constant chosen for $E_0$ at the beginning). 
\end{proof}

\begin{remark} The connected manner thus means that $\Gamma$ on different tori
can be arranged in a sequence by common values of $x$ as above. 
Actually, any patch on a torus can even be replaced by 
a set of two points that do not have to be close to each other: 
to establish the connected chain for $\Phi(x)$, we only 
require two points with different $x$ on one torus, and one of the $x$
shared with a point on another torus so that all $x$ are covered.
\end{remark}

\begin{remark} The lemma can also be 
expanded to concern all potentials (not just integrable
systems) by defining $\Gamma$ to be sections of orbits having common
points $x$. In fact, just one chaotic orbit is sufficient as
it eventually defines $\Phi(x)$ at all $x\in\R^3$. More generally, $\Gamma$
can denote parts of any structures on which $E$ is constant, 
or parts of isosurfaces of any functions of the form $f(E)$. Note that an 
isosurface need not be connected, i.e., $f(E)$ need not
be monotonous. One value of $f$ can correspond to more than one value of 
$E$ as long as the branches of different $E$ can be identified: 
$\Gamma\leftrightarrow E$.\label{remfe}
\end{remark}

The lemma states that even highly fragmentary information on the
shape of the tori in phase space is well sufficient to determine the
integrable potential $\Phi(x)$ uniquely.
Now, let independent distribution functions 
$f_i(x,v)$, $i=1,2,3$, in an integrable potential $\Phi(x)$ 
be defined everywhere in $\R^3\times \R^3$. By the strong
Jeans theorem, the isosurfaces $f_i(x,v)=const$ are of the form
$f[I(x,v)]=const$, where $I(x,v)=const$ 
define the tori created by $\Phi(x)$. The 3-surfaces
formed by the intersection of three 5-surfaces $f_i(I)$ are 3-tori
(defined by $I$ as well). 
Then, by the lemma, any collection of parts of surfaces $f_i=const$ sufficient
to determine a connected chain of torus patches uniquely determines
$\Phi(x)$. We can now state the following uniqueness theorem:

\begin{theorem} Let three independent steady-state 
distribution functions $f_i(x,v)$, $i=1,2,3$ 
($f:\R^3\times \R^3\rightarrow\R$) of 
matter in an integrable system be defined in some common domains of
$\R^3\times \R^3$ such that a set of parts of the surfaces $f_i(x,v)=const$ 
forms a succession of torus patches connected in $x$. Then the $f_i(x,v)$
uniquely determine the potential $\Phi(x)$.\label{theo1}
\end{theorem}

\begin{remark} We can pose this in the integral space
$\R_I^3$ as well: we assume that the $f_i(I)$, $I\in\R^3$, are such that
the set of equations $\lbrace f_i(I)=C_i\rbrace$ has a nonzero and
finite number of
solutions $I$ for sufficiently many $\lbrace C_i\rbrace$ such that
a connected sequence of patches can be constructed
(both $I$ and $f_i(I)$ are, of course, unknown prior to
the determination of $\Phi(x)$). 
Note that the number can be larger than one, i.e.,
$f_i(x,v)$ need not form a proper basis of isolating integrals:
if a set $\lbrace f_i=C_i\rbrace$ corresponds to more than one torus (more
than one set of $\lbrace I_i=C'_i\rbrace$), it is sufficient to be able to
distinguish the different tori in the sense of the lemma (cf.\ remark \ref{remfe} of the
lemma).
\end{remark}

\begin{remark} The theorem emphasizes the role of different 
distribution functions $f_i(x,v)$ in providing information about the system.
If different $f_i$ based on the observations of various steady-state
object populations are
available, we gain more by studying $f_i$ rather than the total 
distribution function $f_{\rm tot}=\sum_i f_i$ of all luminous mass. 
\end{remark}

\begin{remark} Most of all the feasible ways of filling
$\R^3\times \R^3$ with invariant tori, each with a
constant $J$ of  (\ref{action}), are not derivable from any potential
(cf.\ \cite{kaasham}). Thus there usually is no
integrable potential $\phi$ other than $\Phi$
that, with its tori $I_\phi$ and some distribution function $f'$, 
would yield $f'[I_\phi (x,v)]=f[I_\Phi (x,v)]$ everywhere in $\R^3\times \R^3$,
so in most cases one distribution 
function $f(x,v)$ in $\R^3\times \R^3$
determines $\Phi(x)$.\label{remsingle}
\end{remark}

The theorem is constructive for three $f$'s from which
we directly get $\Phi(x)$ with the procedure of the lemma, while the case
of one $f$ (remark \ref{remsingle}) is non-constructive. Finding the integrable
$\Phi(x)$ corresponding to one $f$ is not obvious since
there are no general procedures for finding and exploring integrable 
potentials \cite{licht}. The set of potentials giving rise to 
integrable systems is known to be a vanishingly small subset of all potentials 
\cite{licht}, and there is no such $\epsilon\ne 0$ that,
for an integrable $\phi(x)$, $\phi(x)+\epsilon\varphi(x)$ is still integrable
for an arbitrary $\varphi(x)$. In practice, we can
circumvent this difficulty by allowing the use of non-integrable
potentials and approximate tori, i.e., we construct a set of approximate
integrals $I(x,v)$ for a potential $\Phi(x)$ \cite{kaasprl}; 
this approach will be discussed in section 6.

The importance of isosurfaces in an integrable system
is that they relieve us from having to record
the time evolution of orbits. If the orbital 
tori are definable, we only need to probe phase space at one moment.
It is important to note that $f(x,v)$ does not need to cover all matter in
the system (then the single-$f$ uniqueness 
would be trivial for all systems): it can 
represent any fraction of the matter, even virtually massless test particles,
as long as it is defined on all tori and it has settled to a steady-state
equilibrium. This underlines the fact that we are interested more in 
the geometric structures in $\R^3\times \R^3$ (in practice, in a sufficiently
large domain $\Omega$ of observations) rather than the
actual distribution function $f$ itself.

The fact that we do not need to observe all of the matter in the system to
determine $\Phi(x)$ is of crucial importance. Because of the abundance
of dark matter, a large part of all
matter will not be seen even in ideal conditions.
If, for example, $f$ represents all luminous matter, we
immediately obtain the mass density $\rho_D(x)$ of dark matter by
subtracting (\ref{rho}) from (\ref{poisson}) :
\be
\rho_D(x)=\frac{1}{4\pi G}\nabla^2\Phi(x)-\int_{V_+} \tilde f(x,v)\, d^3 v.
\label{dark}
\ee

Another cause of unobservability is the inevitable fact that 
our instruments are not sensitive enough, or that, e.g., interstellar dust 
prevents us from seeing some regions of space properly. Such effects can
be described by the bias or selection function $\gamma(x,v), \,0<\gamma\le 1$, that
represents the fraction of observable matter of $f(x,v)$ in $d^3 x\, d^3 v$.
The potential $\Phi(x)$ can
be determined uniquely even in the presence of $\gamma(x,v)$ as we
immediately can see if we suppose that we have independent $f_i(x,v)>0$ 
available that share the same $\gamma(x,v)$, since now the ratios of
$f_i$ can be viewed as a new basis of isolating integrals:

\begin{theorem} Let the products 
$g_i(x,v)=\gamma(x,v)\,f_i(x,v)>0$, 
$i=1,\dots,4$, of a bias function $0<\gamma\le 1$ and four  independent
steady-state distribution functions $f_i$ be defined in $\R^3\times\R^3$.
Then the three ratios 
\bee
\frac{g_i(x,v)}{g_1(x,v)}=\frac{f_i(x,v)}{f_1(x,v)},\quad i=2,3,4,
\eee
uniquely determine the potential $\Phi(x)$ as in theorem \ref{theo1}.
\end{theorem}

\begin{remark} As above, this emphasizes the combined information content of 
different distribution functions. Analogously with remark \ref{remsingle}, we can
stipulate for two $f_i$ that in most cases two products $g_i(x,v)=\gamma(x,v)\,f_i(x,v)>0$, 
$i=1,2$, defined everywhere in $\R^3\times\R^3$,
determine $\Phi(x)$ via their ratio $g_2(x,v)/g_1(x,v)=f_2(x,v)/f_1(x,v)$.
\end{remark}

\begin{remark}
The bias function $\gamma(x)$ may depend on $x$ only;
$\gamma: \R_x^3\rightarrow \R$. In such cases
the additional factor $\gamma(x)$
does not really increase the possibility of the existence of another
integrable potential $\phi(x)$ and bias function $\xi(x)$, $0<\xi\le 1$, 
such that $\xi f'=\gamma f$ everywhere. 
This would require there to be an isolating integral
of $\phi(x)$ that everywhere matches an isolating
integral of another potential ($\Phi(x)$) when multiplied with a function of 
$x$ only ($\xi(x)/\gamma(x)$),
while general isolating integrals
are mixed functions of $x,v$, and individual potentials. No examples of
such transformations are known. Thus, usually the single product
$\gamma(x)f(x,v)$ determines the potential $\Phi(x)$ of the
system.\label{remgamx}
\end{remark}

The above results mean that dark matter, selection bias, or fragmentary data
are no fundamental obstacles to dynamical tomography. It should be noted
that, for near-integrable systems in the sense of Kolmogorov-Arnold-Moser
(KAM) theorem \cite{arn}, most of
our results for integrable potentials still hold
to a high precision. This is because they are based on the toroidal
topology of motion in phase space, and almost all motion in near-integrable
systems still occurs on tori that can be viewed as perturbed
versions of the invariant tori of an integrable system. 
Even though the geometric 
structure of orbits in phase space changes radically as we move away
from integrability, allowing chaotic orbits and substructures
of foliated tori around closed orbits, the
transition is smooth dynamically, i.e., motion is still mostly quasitoroidal
within a given resolution \cite{licht}.

The solution of the inverse problem is based on the large number $N$ of 
observed $(x,v)$, from which the distribution function $f(x,v)$ can 
be estimated.
With smaller effective number of phase-space dimensions than six,
i.e., in the case of potential symmetries, $N$ may be sufficient
for actual binning in sufficiently small volumes
$\Delta x\Delta v$ in the reduced phase space. 
Thus we sample $f(x,v)$ directly by counting and can take it
to be our observable, as in section 4. 
For six dimensions, we can, e.g.,  take random samples
of marginal distributions to be our observables, as will be discussed in
section 5. 

\subsection{Self-consistency regularization}

If we assume $\tilde f$ 
to represent all matter in the system, the self-consistency
requirement of $\tilde f$ and 
$\Phi$ can be used as regularization. For example,
we can minimize the discrepancy between (\ref{rho}) and (\ref{poisson}) at each space volume
with the function
\bee
\lambda \int \vert\rho_D(x)\vert\,d^3x,
\eee
where $\lambda$ is a suitable weight factor, and $\rho_D(x)$ is given
by (\ref{dark}). In practice, this can be replaced by
\be
\lambda \sum_i [\rho_D(x_i)]^2,\label{reg0}
\ee
for a large number of test locations $x_i$, written
in a manifest $\chi^2$-form.
As will be discussed below, it is sometimes better to compare potential values
rather than densities
at test points due to numerical instabilities
in evaluating Poisson's equation directly. 

If any amount of dark matter is allowed, no regularization such as 
(\ref{reg0}) can be used if there is no information on the masses
of the stars, i.e., observations correspond to the number
density, and the estimated mass would only be based on the regularization.
This is easy to see by considering the case of constant density fraction of luminous matter:
$\rho_{\rm lum}(x)/(\rho_D+\rho_{\rm lum})=const<1$ 
everywhere in $\R_x^3$.
Regularization of the form (\ref{reg0}) would then always yield exactly
$\rho_D=0$, regardless of the true density ratio. If masses are included
in observations, i.e., we would solve for $\tilde f$ even without
regularization, (\ref{reg0}) can be employed, though minimizing the amount of dark
matter is not necessarily
an appropriate principle then.

\section{Examples in spherical symmetry}

Spherically symmetric systems have the advantage of being integrable, as
now the energy $E(x,v)$ as well as the three components of the angular momentum
$P(x,v)=x\wedge v$ are isolating integrals; the fact that there are four 
integrals instead of three is reflected in that each orbit is constrained 
to a plane (e.g., \cite{bt}). 
The original Jeans' theorem applies only to systems
without such degeneracies, but it can be extended to the spherical case
\cite{lynden},
stating that any steady-state distribution function in a potential
$\Phi(r)$, where $r=\Vert x\Vert$ is the radius, must be of the
form $f(E,P)$. Further, since the system is spherically symmetric, the
direction of $P$ is uniformly distributed over the unit sphere $S^2$, so we
are only interested in functions of the form $f(E,L)$, where
\be
L(x,v)=\Vert x\wedge v\Vert. 
\ee

It is customary (e.g., \cite{bt}) to use the
concepts of relative potential and energy $\Psi$ and $\mathcal E$ such that
\be
\Psi(x):=-\Phi(x)+\Phi_0,\qquad {\mathcal E}(x,v):=\Psi(x)-\frac{1}{2}\langle v,v\rangle
=-E(x,v)+\Phi_0,
\ee
where $\Phi_0$ is chosen such that $f>0$ for ${\mathcal E}>0$
and $f=0$ for ${\mathcal E}\le 0$ (usually, $\Phi_0=0$). We now
study potentials $\Psi(r)$ and distribution functions of the form $f(\E,L)$.
In this section, we simplify the notation by using $f(\E,L)$ 
directly to denote the mass distribution function (i.e., $\tilde f$) 
as it is the form the analytical treatment here always refers to. For
simplicity, let us also assume all the stars to have the same mass $m$.
The systems as well as inversion methods used in this section are
simple ones, as the goal is to investigate the basic properties of the
inverse problem rather than develop a full-fledged analysis machinery.

Our phase space is now essentially reduced to radius $r$
and velocity $v$ in a plane. Because of this reduction in variables,
spherically symmetric cases offer analytical tractability to a much
larger extent than other systems. Of course spherical symmetry rarely occurs
in actual gravitational systems, but it is highly useful for exploring the
basic properties of the inverse problem (uniqueness, stability, the effect
of sampling density, convergence of solution, etc.) in practical computation.
Indeed, because of their analytical prospects, spherical systems 
and distribution functions of the form $f(\E,L)$ have enjoyed
ever-continuing popularity in theoretical dynamics for over a century.

Since now we know the exact functional forms of the isolating integrals, 
we can immediately see that
the $\gamma(x)$-uniqueness of remark \ref{remgamx}
holds exactly for any nontrivial $f(\E,L)$ and $\Psi(r)$. 
All distribution functions $f(x,v)$ now have to be of the form
$f[\E(x,v),L(x,v)]$, and obviously no function $f[\E(x,v),L(x,v)]$
is separable in the form $h(x)g[\E(x,v),L(x,v)]$. Thus there is no
bias function $\xi(x)\ne\gamma(x)$ allowing a product
$\xi(x) f'_{\psi}(x,v)$ to be identical with $\gamma(x) f_{\Psi}(x,v)$
everywhere (note that the bias function $\gamma$ does not have to be 
spherically symmetric). Thus we can state the following uniqueness theorem:

\begin{theorem} If a bias function $\gamma(x)$ 
depends on $x$ only: $\gamma: \R_x^3\rightarrow \R$, the product
$\gamma(x)f(x,v)$, defined everywhere
in $\R^3\times\R^3$ in a spherically symmetric steady-state system, 
uniquely determines $\gamma(x)$.\label{theo3}
\end{theorem}

\begin{remark} Similarly, the restricted form $f[\E(x,v),L(x,v)]$ corroborates
remark \ref{remsingle} algebraically.
\end{remark}

\begin{remark} St\"ackel potentials $\Upsilon(x)$ \cite{bt,dz,gold}
are a class of integrable potentials for which the Hamilton-Jacobi
equation is separable in ellipsoidal coordinates. Just like in the
spherical case, the functional forms of their isolating integrals 
in the distribution function $f[E(x,v),I_2(x,v),I_3(x,v)]$ are known 
exactly \cite{dz}. Examining the forms of the
isolating St\"ackel integrals, we can immediately see that
corresponding uniqueness results 
hold for all steady-state systems governed by
St\"ackel potentials $\Upsilon(x)$.
\end{remark}

The simplest nontrivial distribution functions are
obtained by just removing $L$-dependence and using isotropic $f(\E)$. The
solution for a spherically symmetric system whose $f(\E)$ 
is self-consistent with its
density $\rho(r)$ and the corresponding potential $\Psi(r)$
is now simple to obtain. By writing
(\ref{rho}) in spherical symmetry and changing
the integration variable $v\rightarrow\E$, we obtain an Abel integral
equation for $f(\E)$, the solution of which is the well-known
Eddington's formula \cite{bt}
\be
f(\E)=\frac{1}{\sqrt 8 \pi^2}\Big[\int_0^\E \frac{d^2\rho}{d\Psi^2}
\frac{d\Psi}{\sqrt{\E-\Psi}}+\frac{1}{\sqrt\E}\Big(\frac{d\rho}{d\Psi}
\Big)_{\Psi=0}\Big],\label{eddi}
\ee
where $\Psi$ is used as the argument of $\rho$ via 
$\Psi(r)\rightarrow r(\Psi)$; see \cite{bt} for the conditions for $f(\E)$ to
be non-negative, etc. Note that (\ref{eddi}) can naturally be used
for non-self-consistent distribution functions 
(systems containing dark matter) as well. For example, we 
can solve for the distribution function $f_{\rm lum}(\E)$ of luminous matter
from the luminous fraction of the total density: 
$\rho_{\rm lum}(r)=\varrho(r)\rho(r)$, with a given 
$0<\varrho(r)<1$, thus replacing
$\rho\rightarrow\varrho\rho$ in (\ref{eddi}) while retaining the $\Psi(r)$
corresponding to $\rho(r)$.

With (\ref{eddi}), we can simulate observations of 
the distribution function of a spherically symmetric system,
and then numerically check how well its potential can be recovered.
If self-consistency is used for regularization,
Poisson's equation (\ref{poisson}) 
is simple to integrate, so the regularizing function 
corresponding to (\ref{reg0}), but in terms of potential rather
than density, can be given as
\be
\begin{split}
\chi^2_{\rm reg}&=
\sum_i \Big(\Psi(r_i)-16\pi^2 G[\frac{1}{r_i}\int_0^{r_i}\int_0^
{\sqrt{2\Psi(r')}}f[\E(\Psi(r'),u)]u^2du r'^2 dr'\\
&+\int_{r_i}^\infty\int_0^
{\sqrt{2\Psi(r')}}f[\E(\Psi(r'),u)]u^2du r'^2 dr']\Big)^2\label{regu}
\end{split}
\ee
where $u:=\Vert v\Vert$. 
When using gradient-based methods in optimization, integrals of the form
(\ref{regu}) must be differentiated algorithmically for consistency, 
i.e., by taking
the derivatives of the numerical algorithm for evaluating (\ref{regu})
rather than first taking the analytical derivative of (\ref{regu}) and then
evaluating the integral numerically.

The lemma (remark \ref{remfe}) and theorem \ref{theo3} ensure that there is a unique solution
for both the potential $\Psi(r)$ and the bias function $\gamma(x)$. 
The $\chi^2$-function of our inverse problem can now be represented as
\be
\chi^2=\sum_i\Big[N_im_i-\int_{\Omega_i}\gamma(x)
f[\E(\Psi(r),u)]\,d^3x\,d^3v\Big]^2
+\lambda\chi^2_{\rm reg},\label{chisq}
\ee
where 
$N_i$ is the number of stars observed in the region $\Omega_i$, and $m_i=m$.
To ensure positivity, let us model our potential and distribution function
as
\be
\Psi(r)=\exp\Big(\sum_i a_ir^i\Big),\qquad f(\E)=\exp\Big(\sum_i b_i\E^i\Big).
\label{modfuncs}
\ee
We could include an explicit $1/\E$-term in the exponential sum
for $f(\E)$ to make sure that $f\rightarrow 0$ as $\E\rightarrow 0$, 
but this turns out to be neither necessary nor useful in practice.
Because of the adopted form of $\Psi(r)$, it is especially useful to employ
(\ref{regu}) rather than density regularization as Poisson's equation
would introduce a singularity for one term.

Simple generalizations of the $r^{-1}$-type Kepler potential of a point
mass already yield interesting examples of systems of continuous mass
distribution. In fact, dynamical systems whose potentials 
decrease much faster than this
exhibit instability problems \cite{bt}. This is another reason
why a number of variations of the $r^{-1}$-theme have been developed.
For example, the isochrone potential
\be
\Psi_{IC}(r)=\frac{GM}{b+\sqrt{b^2+r^2}},
\ee
where $M$ is the mass of the system and $b>0$ a scaling constant, with its
density pair $\rho_{IC}(r)$ given by Poisson's equation, yields the
distribution function \cite{bt}
\be
\begin{split}
f_{IC}(\E)&=\frac{M}{(32GMb)^{3/2}}\frac{\sqrt{\tilde\E}}
{(1-\tilde\E)^4}
\Big[27-66\tilde\E+320\tilde\E^2-240\tilde\E^3\\
&+64\tilde\E^4
+3(16\tilde\E^2+28\tilde\E-9)\frac{\arcsin\sqrt{\tilde\E}}
{\sqrt{\tilde\E(1-\tilde\E)}}\Big],
\end{split}
\ee
where $\tilde\E=\E b/(GM)$. We will use this in our numerical examples below.

\begin{figure}
\begin{center}
\includegraphics[width=6.3cm]{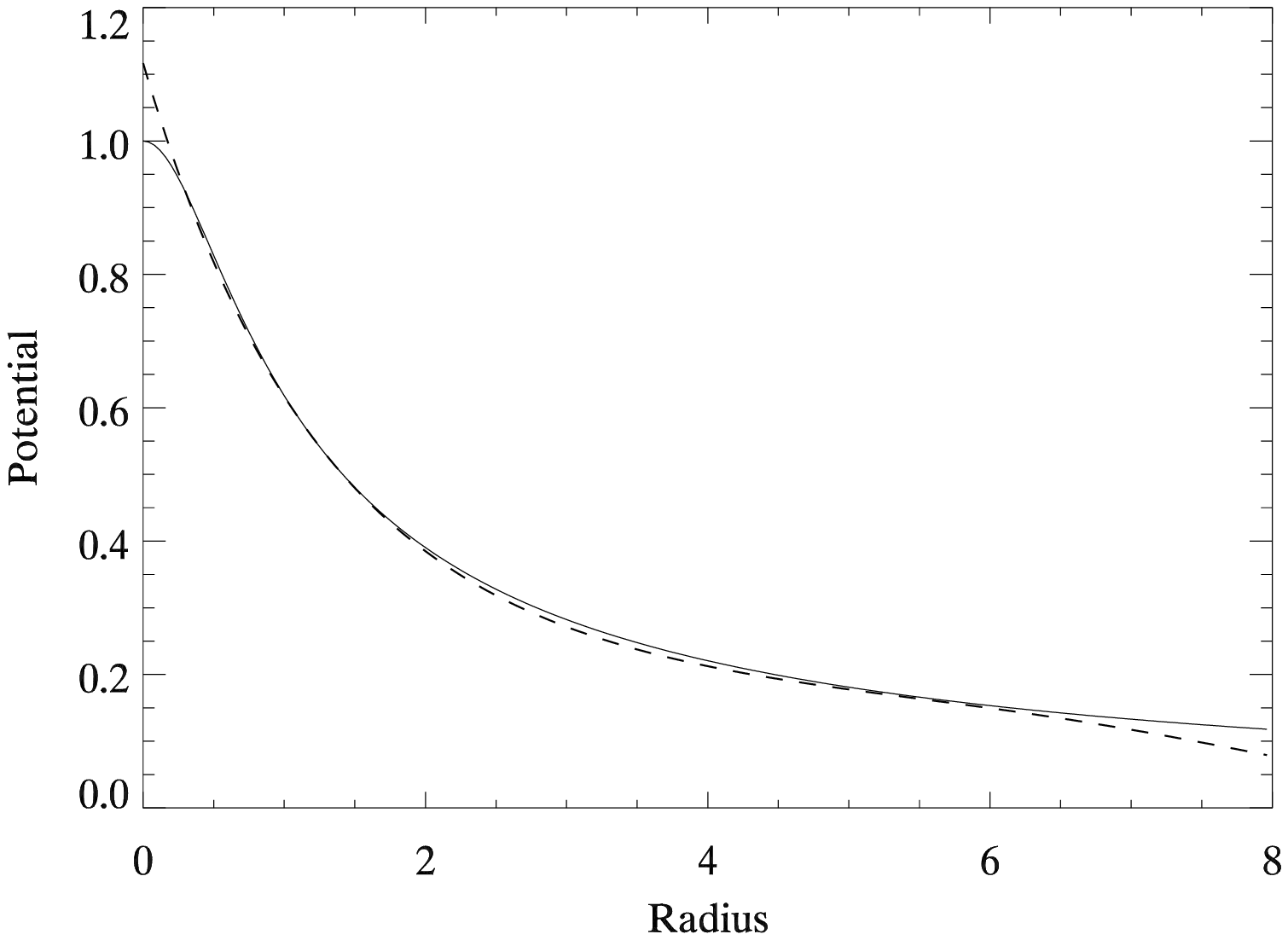}\includegraphics[width=6.3cm]{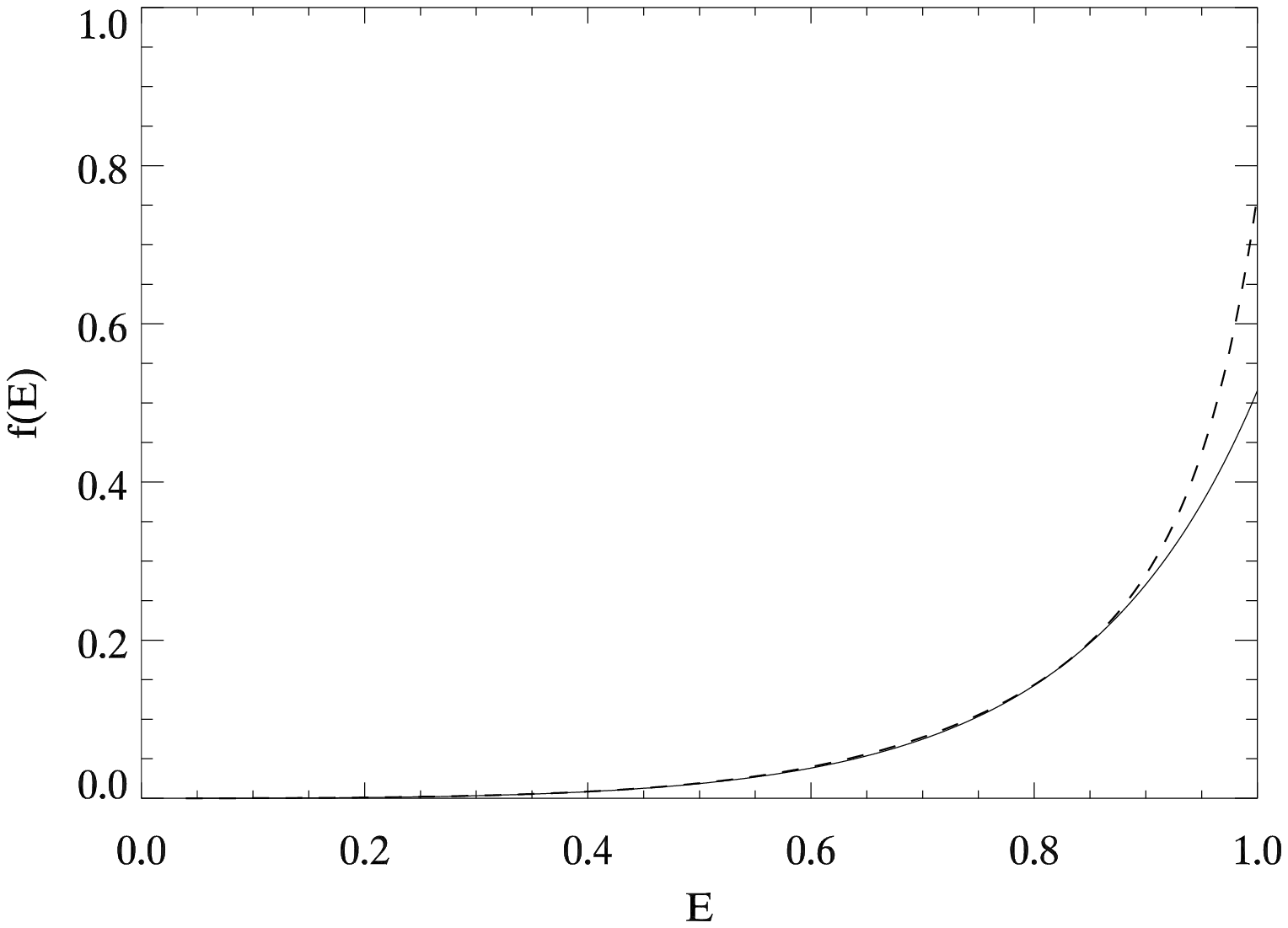}
\caption{(a) $\Psi(r)$ and (b) $f(\E)$ of the isochrone potential 
(solid lines), together with the best-fit models from observed binned
$f(x,v)$ (dashed lines).}
\label{ic}
\end{center}
\end{figure}

We simulated observed $N_i$ by creating 400 randomly distributed $\Omega_i$
in $(r,u)$-space (the extents of each $\Omega_i$
about $1\%$ of the maximum widths
in $r$ and $u$, with a suitable cutoff value $r_{\rm max}\gg b$)
and computing the amount of mass in them from $f_{IC}$. Any 
distribution function corresponding to an $r^{-1}$-type potential drops 
fast as $r$ increases; thus those $\Omega_i$
close to $r=0$ contain most of the mass and dominate $\chi^2$ if
the volumes of $\Omega_i$ are equal. A more sophisticated observation
sampling could use increased volumes away from $r=0$, but in any case
this sampling is artificial (as is the symmetry setup) 
and designed rather to probe the numerical
properties of the problem than to simulate realistic conditions.

The objective
function (\ref{chisq}) was minimized with the Levenberg-Marquardt
grad\-ient-based procedure \cite{press}. 
The initial guess for $\Psi(r)$ was based on the
top velocities at each $r$, and the initial $f(\E)$ simply had a 
correspondingly scaled zeroth-order term and a first-order term
anticipating $f\rightarrow 0$ as $\E\rightarrow 0$. Even with this crude
initial values, the procedure converged well towards the correct minimum,
i.e., $\chi^2$ does not have many local minima far away
from each other in the parameter space at least with reasonably low
truncation degrees of (\ref{modfuncs}). The minimum is not particularly
sharp, i.e., there are many virtually as good solutions ($\chi^2$-level 
varying a few percent) close to each other even in the noiseless case,
as expected because of the insufficiency of the model (``model noise''),
in particular near the centre $r=0$.

The solution was stable to typical error levels of data: adding 5\% or even 10\% noise to the
data values resulted in virtually the same $\Psi(r)$ and $f(\E)$.
An exponential bias function 
$\gamma(x)=\exp (-\Vert x-x_0\Vert/R_0)$ with the centre at different
values of $x_0$
was well solved for by the procedure, together
with $\Psi(r)$ and $f(\E)$. Fig.\ 1 a and b show 
the solution $\Psi(r)$ and $f(\E)$ (dashed lines), 
obtained with coefficients up to degrees 4 and 5, respectively,
against the actual $\Psi_{IC}$ and $f_{IC}$ (solid lines)
with $GM=1$ and $b=0.5$ ($r_{\rm max}=10$), a noise level of $5\%$, 
and a bias $\gamma(x)$ with $R_0=5$ and $\Vert x_0\Vert=2$.
With the bias, the sample cells in $\R_x^3$ were defined by $r$ and the polar
angle $\theta$ ($x_0$ was placed at $r=2$ and $\theta=0$).

Including the regularization term (\ref{regu}) for a self-consistent case 
improved the $\Psi(r)$ fit at large $r$; as a counterexample,
Fig.\ 1 a shows the case without regularization, corresponding
to allowing dark matter; with strong regularization,
the lines essentially coincide at all $r<r_{\rm max}$.
The slight ``overshooting'' of $\Psi(r)$- and $f(\E)$-models at the centre
of the system (at $r=0$ and $\E=\Psi(0)$) is partly an
intrinsic limitation of the polynomial model.

The form $f(\E,L)$ for the distribution function describes the anisotropy
of velocities; there are now infinitely many $f(\E,L)$ corresponding to
a given density $\rho(r)$. Creating general functions of this form even
for spherical symmetries is not as straightforward as Eddington's formula
for $f(\E)$. The Osipkov-Merritt procedure \cite{bt} 
is a standard technique for 
creating a family of self-consistent distribution functions 
with an adjustable anisotropy parameter, but
it is heavily restricted as it still essentially uses Eddington's formula 
by fixing a combination ${\mathcal U}=\E-cL^2$ and stating $f=f({\mathcal U})$.

In our case we can (and must) 
allow the existence of dark matter, so we can carry out the
simulations without the self-consistency requirement. Thus we can simply
use, e.g., 
\be
f(\E,L)=f_0(\E)h(\E,L),\label{fefel}
\ee
where $f_0(\E)$ is any suitable isotropic function, and
the function $0\le h(\E,L)\le 1$ 
can be chosen to represent any desired dynamical
characteristics of the observed matter (anisotropy 
of angular momentum distribution in
various regions, etc.) via some adjustable parameters. The density $\rho(r)$
from this is always $\rho(r)\le\rho_0(r)$, and the difference
$\rho_0-\rho$ is attributed to dark matter; whether or not this interpretation is
realistic is not relevant to this study. 

Suppose we want to describe a system with central orbits mostly radial and 
those reaching high radii mostly tangential. Then we can choose, for example,
\be
h(\E,L)=\frac{1}{1+d}\Big\lbrace 1-d\cos\Big[\pi\Big(\frac{L^2}{L_m^2(\E)}+\frac{\E}{\E_m}\Big)\Big]\Big\rbrace,
\label{hel}
\ee
where $0\le d\le 1$ is an amplitude factor, $\E_m$ is the maximal $\E$, and 
the maximal angular momentum $L_m=r\,u(r,\E)$ occurs for given $\E$ at 
$\hat r$ for which $d L_m(\hat r)/dr=0$, i.e., $L_m(\E)$ is obtained from
\be
2\E=\hat r\frac{d\Psi_0(\hat r)}{dr}+2\Psi_0(\hat r)\Rightarrow \hat r,
\qquad L_m(\E)=\hat r\sqrt{2(\Psi_0(\hat r)-\E)},
\ee
where $\Psi_0(r)$ corresponds to $\rho_0(r)$.

\begin{figure}
\begin{center}
\includegraphics[width=6.3cm]{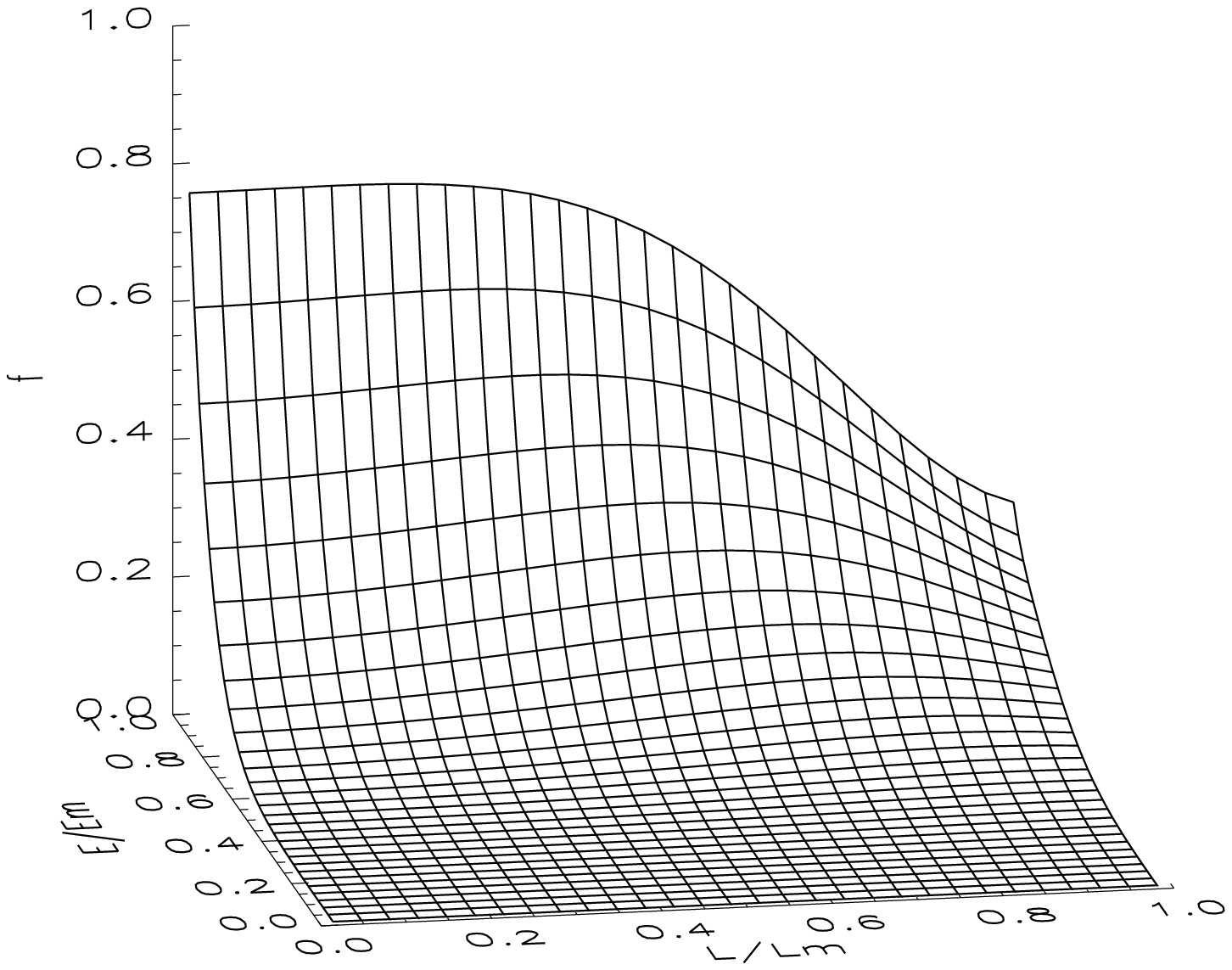}\includegraphics[width=6.3cm]{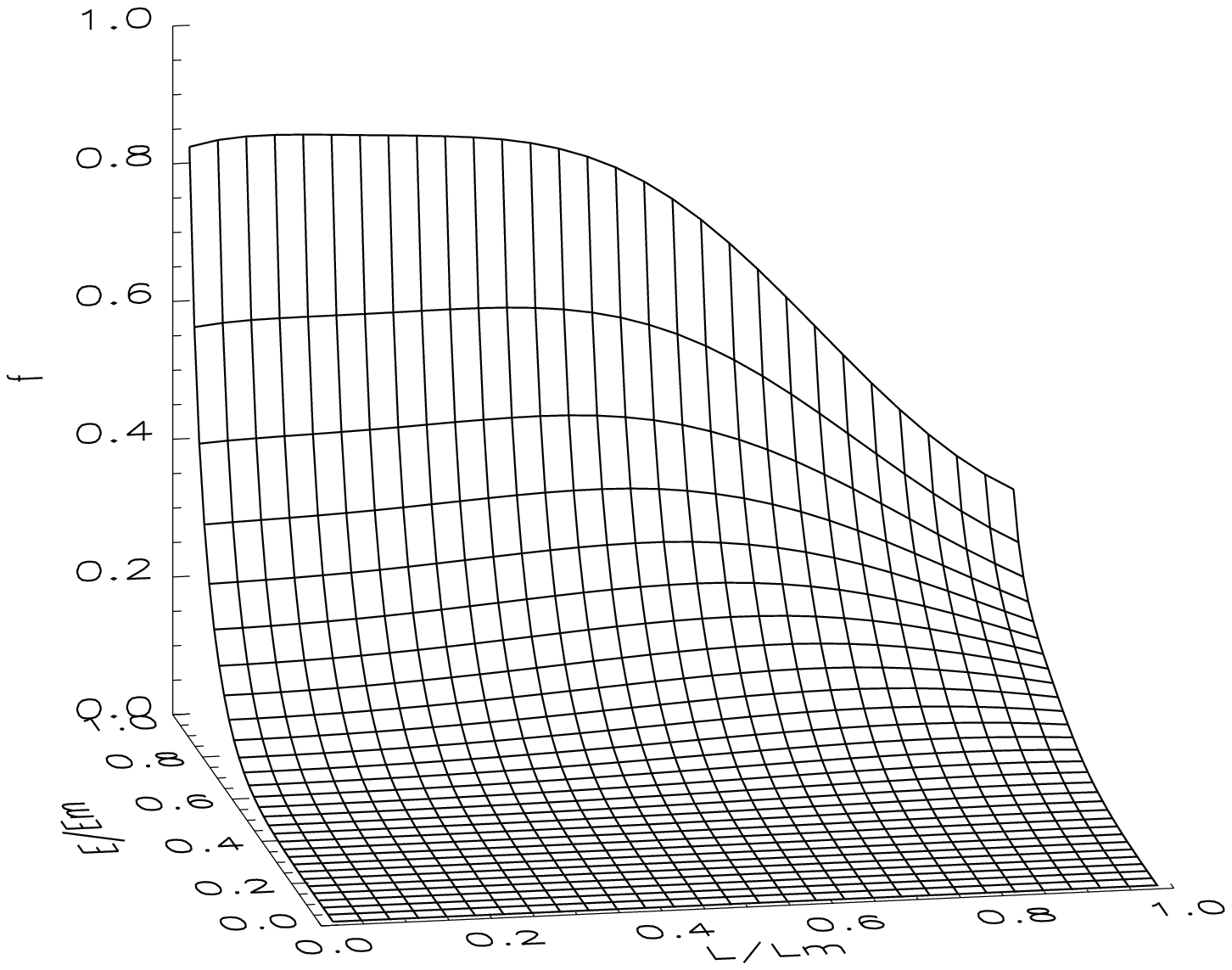}
\caption{(a) A chosen $f(\E,L)$ for the isochrone, and (b) the
best-fit model. The plots are in $(\E /\E_m,L/L_m(\E)$-space.}
\label{ic2}
\end{center}
\end{figure}

Using the isochrone for $f_0(\E)=f_{IC}(\E)$ and setting $G=1$, we have 
\be
L_m^2(\E)_{IC}=Mb\Big(\frac{M}{2\E b}+\frac{2\E b}{M}-2\Big).\label{lm2}
\ee
Fig.\ 2a shows the (suitably normalized) 
$f(\E,L)$ from $f_{IC}(\E)$ (as with the
example of Fig.\ 1) and (\ref{hel}) with an amplitude $d=1/2$,
plotted in the $(\E/\E_m,L/L_m(\E))$-space for convenience. The plotted
surface corresponds to $[0.05,0.95]\times[0.05,0.95]$ so that the behaviour 
of $f(\E,L)$ can well be seen without the shrinking $L$-ranges due to 
(\ref{lm2}). The inverse solution, based on 4000 bins
distributed in $(r,v)$-space in
the same manner as in the example of Fig.\ 1, and with the same computational
procedure for solving the inverse problem, is shown in Fig.\ 2b.
The $(r,v)$-space is here given by the dimensions $(r,v_r,v_\phi)$, 
where $v_r=dr/dt$ is the radial velocity and the tangential velocity is
$v_\phi=\sqrt{u^2-v_r^2}$, i.e., $L=rv_\phi$.

With the dark-matter scenario, we have no self-consistency regularization 
(and $m=1$ in the $f(\E,L)$-version of (\ref{chisq})). 
The model for $f(\E,L)$ was taken to be
\be
f(\E,L)=\exp\big(\sum_{ij} b_{ij}\,\E^iL^j \Big),\label{felmodel}
\ee
with the largest degree of 5 for both $\E^i$ and $L^j$. The two-dimensional
polynomial series may not be the best general representation for
$f(\E,L)$ (cf.\ the intrinsic limitations mentioned with Fig.\ 1), 
but in this case it at least appears to lead essentially 
to as good a convergence 
and solution as in the one-dimensional polynomial case.

\section{Distribution function as probability distribution}

To be able to solve the inverse problem in terms of probability distributions
and without binned sampling,
we need to define how to compare a model distribution and observations, i.e.,
a sampled realization of some probability distribution. Direct comparison
methods of distributions are based on cumulative distributions as the latter
are well-defined non-binned observables. However,
cumulative distributions are uniquely defined only in one dimension
though there have been some attempts at two and three dimensions; 
see \cite{press} and references therein. This dilemma is solved
by noting that the one-dimensional marginal distributions 
of our $f(x,v)$ exactly correspond
to the usual line projections in the 
projection-slice theorem, Radon transform 
and related tomographic analysis; see, e.g., \cite{brace,kaipio} and references therein.

In general terms, let $f(x)$, $x\in\R^N$, be a probability distribution,
and let $\sf R$, an $N\times N$ orthogonal rotation matrix, ${\rm det}\,{\sf R}=1$, 
describe a linear 
transformation to a new basis in $\R^N$: $w={\sf R}x$. 
The marginal distribution
function $h(z)$ along $z$, any one of the new coordinate axes denoted
by the set $W({\sf R})$, is
\be
h(z)=\int_{W({\sf R})\setminus z} f({\sf R}^{-1}w)\,d^{N-1}w,
\ee
and its cumulative distribution function $C(z)$ is
\be
C(z)=\int_{z_{\rm min}}^z h(z')\,dz',
\ee
usually with $z_{\rm min}\rightarrow -\infty$.
Since $h(z)=dC(z)/dz$ uniquely defines $h(z)$ from a given $C(z)$, 
we know from tomographic theory (hence
this inverse problem can be called dynamical tomography) that:

{\em The probability
distribution $f(x)$, $x\in\R^N$, is uniquely determined by the
cumulative distributions $C(z)$ of its marginal distributions
$h(z)$ along all line directions in $\R^N$ that define the coordinate axis directions of $z$.}
For example, in $\R^2$ the lines are defined by all possible values of the
rotation angle
$\theta\in S^1$, while in $\R^3$ they are given by all the directions
$\omega\in S^2$.

Two one-dimensional distributions can be compared via their cumulative
distribution functions. In the case of observations and a model,
we denote the observational distribution of
$K$ observations at $z_i$ (arranged in ascending order by $z$) by
\be
S_K(z)=i, \quad z_i\le z<z_{i+1},
\ee
for number density, or, if mass is included in our problem,
\be
S_K(z)=\sum_{j=1}^i m_j,\quad z_i\le z<z_{i+1},
\ee
with $S_K(z)=0$ for $z<z_1$.
Now we are using cumulative distributions $S_K$ as our observables. 
The comparison between $S_K(z)$ and a model $C(z)$ can be done with a number 
of norms; with the usual $L_2$ we have $\chi^2$-comparison. Another
choice often used is the $L_\infty$-norm giving simply
\bee
\mathfrak{D}=\max_{-\infty<z<\infty}\vert S_K(z)-C(z)\vert,
\eee
with which one can use the Kolmogorov-Smirnov (K-S)  
probability $0\le P\le 1$
of matching marginal distributions \cite{mises,press}, when
both $S_K(z)$ and $C(z)$ are normalized to unity by
\bee
S_K(z)\rightarrow\frac{S_K(z)}{S_K(z_K)},\quad 
C(z)\rightarrow\frac{C(z)}{C(z_{\rm max})},
\eee
usually with $z_{\rm max}\rightarrow\infty$,
and one normalizes the obtained $f$ afterwards such that $\int f\,d^N x=K$.
This, however, distorts the comparison 
as the high-$z$ ends are now always matched at unity at the comparison stage.
Also, K-S probability $P$ is defined by
a series with each term a power of
\bee
d_{KS}=\exp(-K \mathfrak{D}^2),
\eee
which means that, for a large number of observations $K$, the model $C(z)$
must be very close to the observed $S_K(z)$ ($K\mathfrak{D}^2<10)$ for the formal
probability $P(d_{KS})$ to have any computationally useful level above zero. 
Even if the observations were perfect, the model cannot reach such a fit,
so usually $P(d_{KS})\rightarrow 0$ regardless of the model.

It is thus practicable to use the $\chi^2$ from the comparison pairs
at each $z_i$; if $K$
is very large, the comparison can be done for a smaller set of points in
each $z$ via pruning or interpolation as the curves of $S_K(z)$ are now so smooth
that this loses virtually no information. Thus we can measure the 
$\chi^2$-sum of marginal distribution comparisons at each
chosen coordinate line direction $i$ and its points $z_j$:
\be
\chi^2=\sum_{ij} [S_K^{(i)}(z_j)-C^{(i)}(z_j)]^2.\label{probchi}
\ee

In principle, we could use an arbitrarily large number of projection 
lines $z^{(i)}$ as, in contrast to standard tomography,
there are no limiting factors for choosing them (except for computation 
time). However, we can expect that even a very
limited number of projection lines is compensated
for by the strong a priori information 
(or assumptions) on $f(x,v)$. As is known
from, e.g., limited-angle tomography, incorporating
even simple a priori information 
is a powerful means of enhancing the solution of the inverse problem
\cite{kaipio}.

We can, for example, expect that simple rotations in each of the
15 two-dimensional coordinate planes in $\R_x^3\times\R_v^3$ are sufficient for
forming a set of projection lines. Now two rows of the 6$\times$6 coordinate 
transformation matrix $\sf R$, defining two choices for $z$, correspond 
to the rotation through 
some angle $\theta$, while for other rows $\sf R$ is an identity matrix.
We evaluate each
$h_i(z)$ corresponding to the choice of $z$ ordered by $i$ using $w={\sf R}_i (x,v)^T$: 
\be
h_i(z)=\int \gamma({\sf R}_i^{-1}w)\, f({\sf R}_i^{-1}w) \,d^5w,
\ee
and find, via $S_K^{(i)}(z)$ and $C^{(i)}(z)$, 
the best model parameters of $\Phi(x)$, $f[I_\Phi(x,v)]$, and
$\gamma(x)$ minimizing (\ref{probchi}). 

We can use the $\chi^2$-formalism to examine different parts of
$\R_x^3\times\R_v^3$ with different weights, using
$[z_{\rm min},z_{\rm max}]$ other than $]-\infty,\infty[$. 
For example, we can zoom in
on parts far away from the centre and fit their cumulative distributions
in separate $\chi^2$-terms -- 
otherwise their effect on $\chi^2$ would be negligible
compare to centre parts. Also, with marginal distributions we
can even use data lacking some coordinates or having poor accuracy in 
them. For example, if some observations have only recorded $x$ without
$v$, we can still use them in marginal distributions in $x$.

While a one-dimensional probability distribution is easy to sample
using the inverse function of its cumulative distribution, creating 
a corresponding multidimensional
coordinate transform $x\leftrightarrow y$ , $x,y\in\R^N$,
such that uniform random sampling in $y$ would automatically 
create a correct distribution in $x$ is difficult and not always realizable. 
Thus we have to resort to a selective Monte Carlo sampling algorithm in
simulations:

\smallskip
{\bf Sprinkler algorithm for Monte Carlo sampling an 
$N$-dimensional distribution $f(x)$, $f:\R^N\rightarrow\R$.}

\begin{enumerate}
\item Draw a random point $x\in\Omega$ (uniform distribution
in $\R^N$), where $\Omega$
is the desired region $\Omega\subset\R^N$.
\item Evaluate the probability
\bee
p(x)=\int_{\omega(x)} f(x') d^N x',
\eee
where $\omega(x)\subset\Omega$ is a finite but small integration region
around $x$, similar (with the same volume) for each drawn $x$.
\item Draw a random number $0\le q\le 1$ (uniform distribution
for $q\in\R$). The point
$x$ is included in the sample set $\mathcal S$ 
if $q\le p(x)/p_m$, where $p_m=\max p(x)$, $x\in\Omega$.
\item Return to 1 until $\mathcal S$ is 
large enough to approximate $f(x)$.
\end{enumerate}

As $N$ increases, the ratio between included and drawn $x$ decreases for peaked distributions, so for $N>3$ this method
gets slower than for low dimensions; in such cases faster Monte Carlo sampling can be carried out
with Markov chain based MCMC tehniques \cite{kaipio}.

\begin{figure}
\begin{center}
\includegraphics[width=6.3cm]{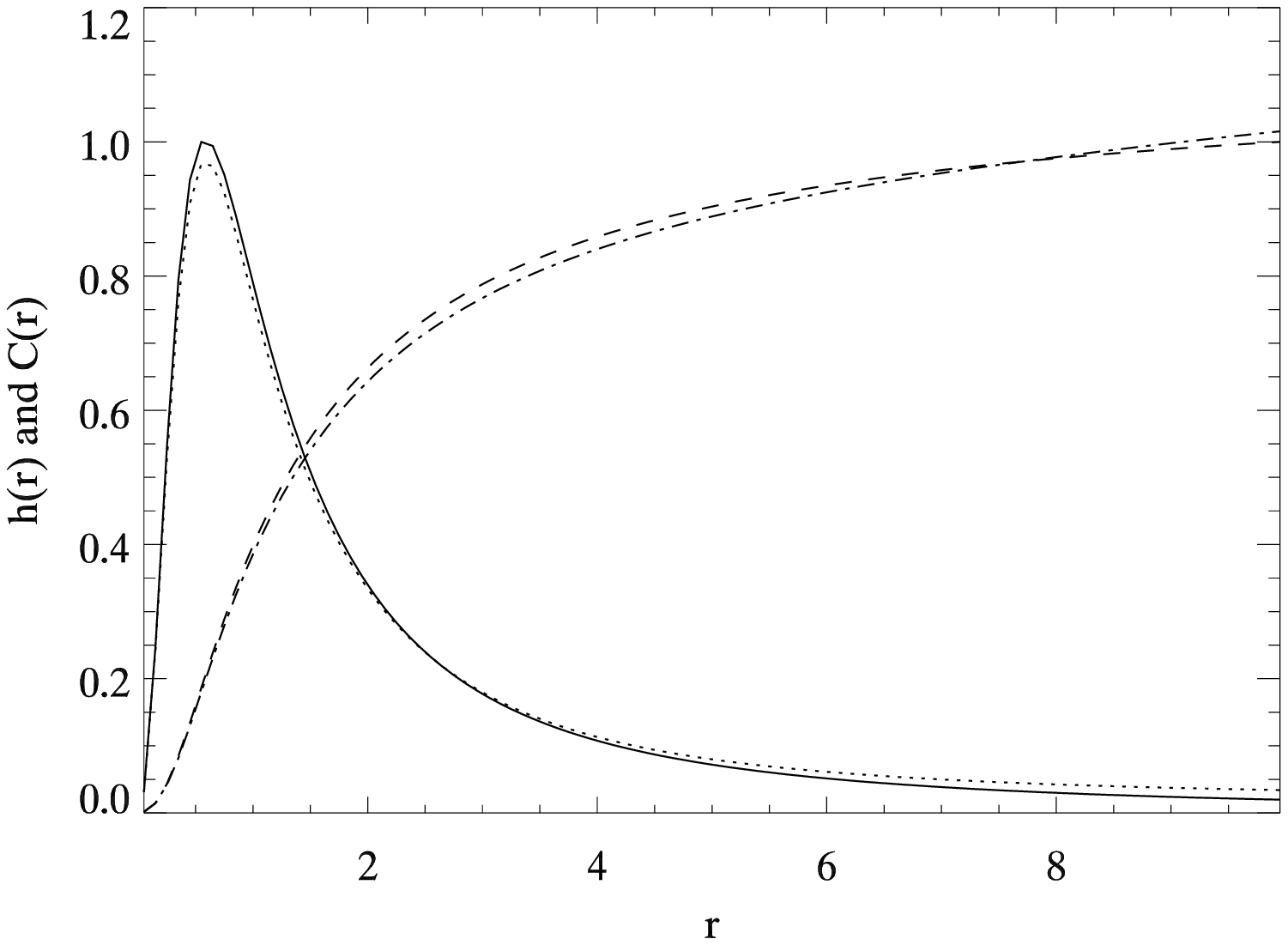}\includegraphics[width=6.3cm]{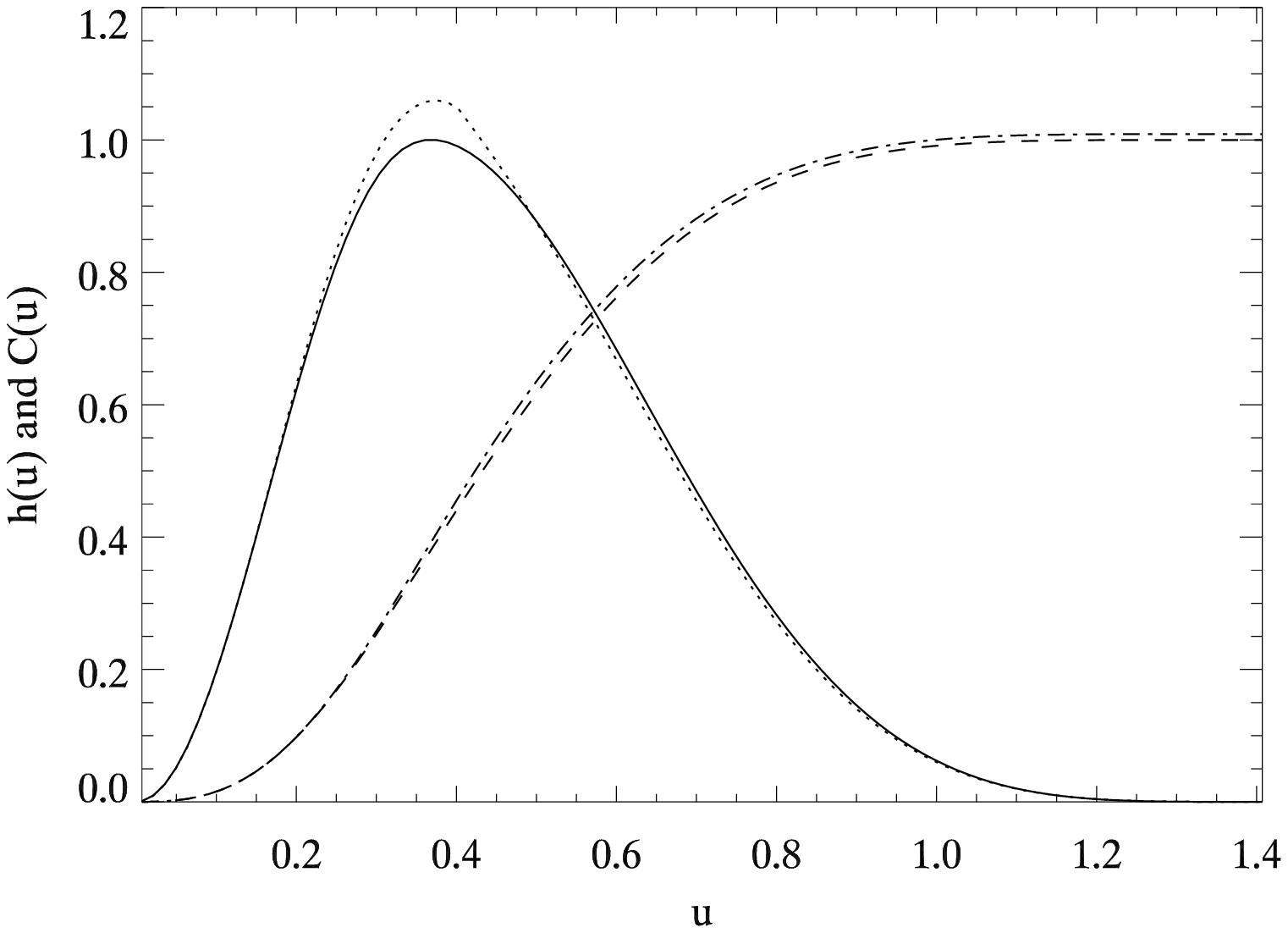}
\caption{Marginal and cumulative distributions in (a) $r$ and (b) $u$ 
from the isochrone distribution ($h$ in solid line, $S_K$ from samples in
dashed line), and the distributions with the best-fit model ($h$ in dotted
line, $C$ in dot-dash).}
\label{cumu}
\end{center}
\end{figure}

As an example, we created $10^5$ random samples of the 
$f_{IC}(\E)$ of the previous examples
in the $(r,u)$-plane in the above manner.
The distribution function $\hat f(r,u)$ in the $(r,u)$-plane is given by
\be
\hat f(r,u)=16\pi^2 r^2\,u^2\,f[\mathcal{E}(r,u)],
\ee
and its marginal distributions in $r$ and $u$ are shown as solid lines
in Fig.\ 3 a and b, suitably normalized to scale between 0 and 1.
The cumulative distributions $S_K$ in $r$ and $u$ 
from the random samples are shown as dashed lines (again suitably normalized).
The best-fit model solution $C$ to (\ref{probchi}), with $\Psi(r)$ and
$f(\E)$ of (\ref{modfuncs}) as earlier, matches these well (dot-dash);
its marginal distributions are shown as dotted lines. The solution is
virtually the same as in Fig.\ 1, so the tomography of the 
probability distribution is a practical approach. In particular, just
the two distributions $S_K(r)$ and $S_K(u)$ were already sufficient for
obtaining the solution: due to the highly restricted form of $f(x,v)$, no
more projection lines were needed.

\section{Conclusions and discussion}

We have defined the inverse problem of dynamical tomography and shown
that, with suitable assumptions and mathematical
tools, the problem is well-posed and solvable.
The uniqueness theorems are central to dynamical tomography: they demonstrate
that the steady-state assumption allows a unique solution even with 
fragmentary data, unobservable mass, and observational bias functions. 

Another important concept is the possibility to approximate general
(or at least near-integrable) systems with integrable ones.
Near-integrable systems that are not very old, but old enough to
have settled to a quasistable state, appear closer to integrable
than old ones as phase-space diffusion (Arnold diffusion) is not extensive 
yet. As Nekhoroshev's theorem as well as semianalytical approximations
and numerical estimates show \cite{licht},
the time scale for such diffusion grows fast in inverse proportion to
the distance of the initial phase-space point from a KAM torus.
For example, in \cite{kaaspre} it was shown that even in chaotic
zones it is possible to construct approximate tori such that
orbits of the system resemble perturbed motion on these tori 
for small enough time intervals.

Torus construction \cite{kaasprl}
is a well-defined concept for near-integrable potentials \cite{poschel}: 
if a Hamiltonian system is 
near-integrable in the sense of the KAM theorem, it is integrable on a Cantor
set consisting of the surviving KAM tori of the system, i.e., it is possible
to construct tori such that their defining action integrals indeed 
parametrize a global, ordered set. 
Via torus construction, we find a potential $\Phi(x)$
for which we can create an approximate set
of tori (and determine expressions for their action integrals $J$
and thus for distribution functions $f[J(\Phi(x);x,v)]$)
defining an integrable system such that the steady-state
integrability assumption used here agrees with the observations as well
as possible. Note that $\Phi(x)$ itself does not have to be integrable or even 
near-integrable; the torus set constructed defines an integrable
Hamiltonian that is not usually derivable from a potential, so we never
get an approximate integrable potential in the first place. This gives
the derived $\Phi(x)$ some additional flexibility; indeed, any integrable
potential approximating a real galaxy is probably not a very good
representation. The key principle allowing the flexibility is the same
as in the uniqueness theorems: we look for isosurfaces in phase space
best explaining the observations. The goodness of our solution $\Phi(x)$
is not directly measured by its closeness to an integrable system. 

Our $\Phi(x)$ should thus be able to mimic the real potential quite well
in cases of near-integrable or even somewhat chaotic but not very old 
systems. The feasible potentials $\Phi(x)$ that can reproduce, via the 
torus-construction principle, the approximate
isosurface structure of the observed matter distribution can all be expected
to be close to each other. In fact, even if $\Phi(x)$ were integrable, we 
would have to use torus construction to find it even if we started with
an integrable initial $\Phi_0(x)$, since the iteration procedure
(or equivalent) for fitting a model to observations will generally
explore non-integrable potentials $\Phi(x)$.

Let us denote by $H_0$ the integrable Hamiltonian corresponding to the
tori constructed for $\Phi(x)$, and by $H$ the Hamiltonian of $\Phi(x)$.
For optimal $H_0$, the difference (in some chosen norm)
\bee
\Vert \delta H\Vert
=\Vert H-H_0\Vert
\eee
is as small as possible over the whole phase space; 
we call the corresponding tori
the optimal tori of $\Phi$, and $H_0$ the optimal Hamiltonian of $\Phi$.
If $\Phi$ is integrable, then its optimal tori are its invariant tori: 
$\delta H$ vanishes everywhere. For our purposes, 
$\Vert\delta H\Vert$ 
does not have to be small (although the smaller it is, the better). 

We surmise that: 
\begin{enumerate}[i)]
\item The optimal tori  
constructed for $\Phi(x)$ form a map
$\Phi\rightarrow H_0$. With the same construction scheme, a potential
$\Phi(x)+\epsilon\phi(x)$ yields an optimal Hamiltonian $H_0+\epsilon h_0$.
\item The isosurfaces of the distribution function of the system can 
be approximated by constructing them from a set of 3-tori describable 
by an integrable Hamiltonian $H_0$ (but not necesssarily by an integrable 
potential).
\end{enumerate}

Then our $\Phi$ can be expected to be a good approximation of the 
real potential of the system (as a regularizing constraint, we can
choose, e.g., the smallness of $\Vert \delta H\Vert$). 

The possibility of constructing optimal tori and Hamiltonians holds another 
great advantage: a multitude of deviations from integrability can
readily be modelled with a suitably tailored version of 
Hamiltonian perturbation theory \cite{kaasham}. This includes
both structural detail (resonant orbit families and other details)
and timelike irregularities (deviations from steady state).

In a forthcoming study, we will study
more general and realistic systems such as axisymmetric and fully 
three-dimensional St\"ackel potentials. We will also introduce more
general observational biases and other factors, and investigate their
influence. The final goal is the combination of a general torus-construction
machinery and an analysis procedure for dynamical tomography, so that we
can work with any potentials. With such tools, we can analyze the data from
large-scale surveys and construct a consistent mathematical model of the
dynamics of our galaxy.  The real galactic problem can be expected to suffer
from considerable model noise, so we should employ various forms of
modelling the distribution functions and potential. In addition to mapping
the density distribution of the dark matter in our galaxy, we should also be able
to test whether distributions with alternative theories of gravity 
are distinguishable from standard models with dark matter.

\section*{Acknowledgements}

It is a pleasure to thank Teemu Laakso, Emiliano de Simone, 
Lassi P\"aiv\"arinta, Jari Kaipio, and Antti Kupiainen for discussions and comments.

\medskip
Received April 2008; revised August 2008.

\medskip
{\it E-mail address:} Mikko.Kaasalainen[at]helsinki.fi

\end{document}